\documentclass[runningheads,a4paper]{llncs}
\usepackage{graphicx,amssymb,amsmath}

\title{Computing Boundary Cycle of a Pseudo-Triangle Polygon from its Visibility Graph}

\author{
	Hossein Boomari\thanks{{\tt h.boomari1@student.sharif.ir}}
    \and
	Soheila Farokhi\thanks{{\tt soheilafar.2011@gmail.com}}
}

\index{Boomari, Hossein}
\index{Farokhi, Soheila}

\newtheorem{obs}{Observation}
\newtheorem{observation}[obs]{Observation}

\authorrunning{Topics in Theoretical Computer Science 2020}

\institute{Department of Mathematical Sciences,\\Sharif University of Technology}

\begin{document}
\thispagestyle{empty}
\maketitle

\begin{abstract}
Visibility graph of a simple polygon is a graph with the same vertex set in which there is an edge between a pair of vertices if and only if the segment through them lies completely inside the polygon. Each pair of adjacent vertices on the boundary of the polygon are assumed to be visible. Therefore, the visibility graph of each polygon always contains its boundary edges. This implies that we have always a Hamiltonian cycle in a visibility graph which determines the order of vertices on the boundary of the corresponding polygon. In this paper, we propose a polynomial time algorithm for determining such a Hamiltonian cycle for a pseudo-triangle polygon from its visibility graph. 
\end{abstract}

\section{Introduction}
\label{sec:A}
Computing the visibility graph of a given simple polygon has many applications in computer graphics \cite{CG}, computational geometry \cite{ghosh-book} and robotics \cite{robot}. There are several efficient polynomial time algorithms for this problem \cite{ghosh-book}.

This concept has been studied in reverse as well: Is there any simple polygon whose visibility graph is isomorphic to a given graph and if there is such a polygon, is there any way to reconstruct it (finding positions for its vertices on the plain)? The former problem is known as recognizing visibility graphs and the latter one is known as reconstructing polygon from visibility graph. Both these problems are widely open. The only known result about the computational complexity of these problems is that they belong to \textit{PSPACE} \cite{everet-thesis} complexity class and more precisely belong to the class of \textit{Existence theory of reals} \cite{exist}. This means that it is not even known whether these problems are \textit{NP-Complete} or can be solved in polynomial time. Even, if we are given the Hamiltonian cycle of the visibility graph which determines the order of vertices on the boundary of the target polygon, the exact complexity class of these polygons are still unknown.

As primitive results, these problems have been solved efficiently for special cases of tower, spiral and pseudo-triangle polygons. A tower polygon consists of two concave chains on its boundary which share one vertex and their other end points are connected by a segment (See Fig.~\ref{fig:inst0}.a). A spiral polygon has exactly one concave and one convex chain on its boundary (See Fig.~\ref{fig:inst0}.b). The boundary of a pseudo-triangle polygon is only composed of three concave chains. The recognizing and reconstruction problems have been solved for tower polygons \cite{tower}, spiral polygons~\cite{spiral}, and pseudo-triangle polygons~\cite{pseudotriangle} in linear time in terms of the size of the graph. The algorithms proposed for realization and reconstruction of spiral polygon and tower polygons first find the corresponding Hamiltonian cycle of the boundary of the target polygon and then reconstruct such a polygon(if it is possible). But, the proposed algorithm for pseudo-triangle polygons needs the Hamiltonian cycle to be given as input as well as the visibility graph, and, having this pair reconstruct the target pseudo-triangle polygon. We use pseudo-triangle instead of pseudo-triangle polygon in the rest of this paper.
\begin{figure}[ht]
\centerline{\includegraphics[scale=0.5]{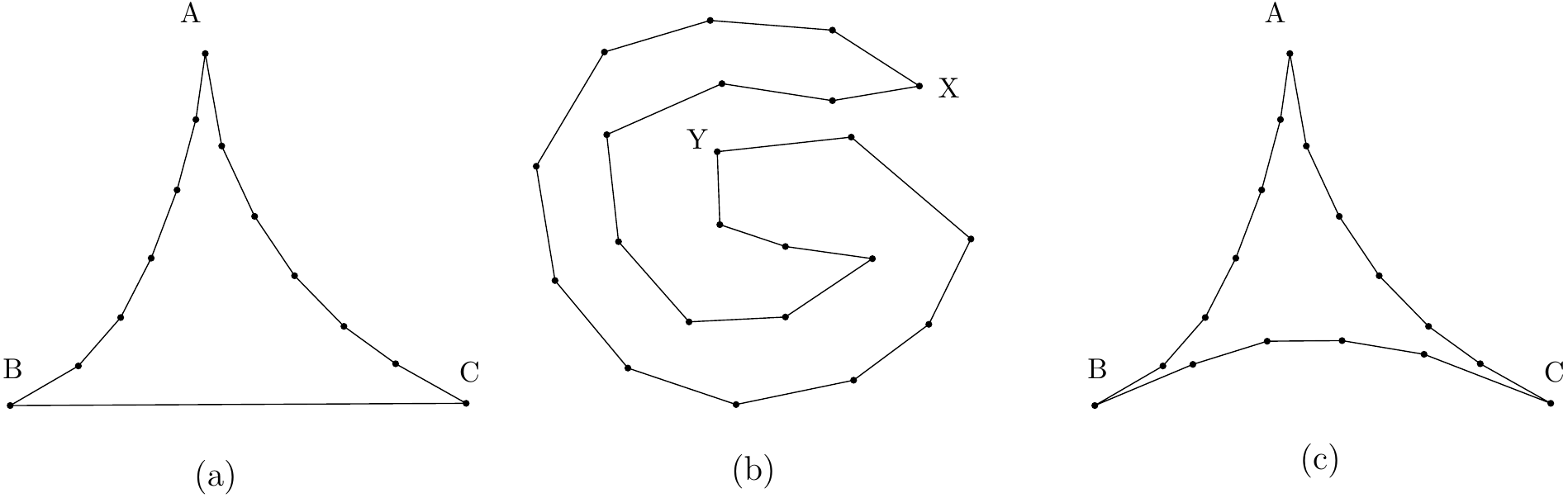}}
\vspace*{8pt}
\caption{a) Tower polygon, b) Spiral polygon, c) Pseudo-triangle polygon.\label{fig:inst0} }
\end{figure}

In this paper, we propose a method to find a Hamiltonian cycle of a realizable pseudo-triangle from its visibility graph in polynomial time. Therefore, the result of this paper in companion with the reconstruction method for realizing pseudo-triangles in \cite{pseudotriangle} will solve the realization and reconstruction problems for a pseudo-triangles from its visibility graph.

In the rest of this paper, we first review the algorithm of solving recognition problem for tower polygons and give some notations, definitions and properties of pseudo-triangles to be used in next sections.

\section{Preliminaries and definitions}
With a given pair of visibility graph and Hamiltonian cycle, Colley et al. proposed an efficient method to solve recognizing and reconstruction problems for tower polygons\cite{tower}.  Here, we review their method, briefly.

A graph is the visibility graph of a tower polygon if and only if by removing the edges of the Hamiltonian cycle from the graph, an isolated vertex and a connected bipartite graph are obtained and the bipartite graph has \textit{strong ordering} following the order of vertices in the Hamiltonian cycle. A strong ordering on a bipartite graph $G(V,E)$ with partitions $U$ and $W$ is a pair of $<_{U}$ and $<_{W}$ orderings on respectively $U$ and $W$ such that if $u<_{U}u^{\prime}$, $w<_{W}w^{\prime}$, and there are edges $(u, w^{\prime})$ and $(u^{\prime},w)$ in $E$, the edges $(u^{\prime},w^{\prime})$ and $(u,w)$ also exist in $E$. Graphs with strong ordering are also called \textit{strong permutation graphs}.

\subsection{Leveling a tower polygon}
The algorithm proposed by Colley et al. for reconstruction of a tower polygon introduced a method named \textit{levelling} for visibility graph of tower polygons. In this method the set of vertices of a tower polygon is covered with some subsets of its vertices called levels. The induced graph on each subset is a clique and each level is labeled with a number. It starts with level $l_1$ which contains the top vertex of the polygon. There are at most two candidates for the top vertex of a tower polygon. The details of the leveling method is given in Section~\ref{sec:leveling}. An assignment of vertices of $G$ to the chains is called a \textit{bordering}.

\subsection{Leveling Method}
\label{sec:leveling}
\begin{observation}
In $G(V, E)$ the degree of the top vertex is 2 and there are at least one and at most two vertices of degree 2. Therefore, there are at most 2 candidate for the top vertex of a tower polygon. In case there is 2 candidate for this vertex the other one is on the bottom of the tower.
\end{observation}

Level $l_2$ contains the neighbours of the top vertex. There is a set ($l^\prime_{i+1}$) of at least one and at most two vertices in $V - \bigcup_{j=1}^{i} {l_j}$ which make a clique with vertices of $l_i$. $l_{i+1}$ contains:
\begin{enumerate}
\item If these vertices are the last vertices of $V$, which are not in any levels, they make the last level ($l_k$).
\item Otherwise, if there is two vertices in $l^\prime_{i+1}$, then $l_{i+1} = l^\prime_{i+1}$.
\item In the last case, there is a single vertex ($p$) in $l^\prime_{i+1}$ and this vertex with one of the vertices of $l_i$ makes $l_{i+1}$. In these situations exactly one of the vertices of $l_i$ has a neighbour in $V - \bigcup_{j=1}^{i+1} {l_j}$. This vertex of $l_i$ and $p$ makes $l_{i+1}$.
\end{enumerate}

Starting with a top vertex, there is a single leveling for the vertices of a tower polygon. The starting level contains the top vertex, the last level contains one or two vertices. If there are two vertices in the last level, they make the base of the tower and in case there is one vertex in the last level, the degree of this vertex is 2. Each of the other levels (middle levels) contains two vertices. The graph $G^\prime$ is the leveling graph of a visibility graph $G(V,E)$ which its vertex set is $V$ and its edges are those edges in $E$ which has not the endpoints in two consecutive level (see Figure~\ref{fig:tower_border}). We have the following observations about leveling and $G^\prime$.

\begin{figure}[ht]
\centerline{\includegraphics[scale=0.5]{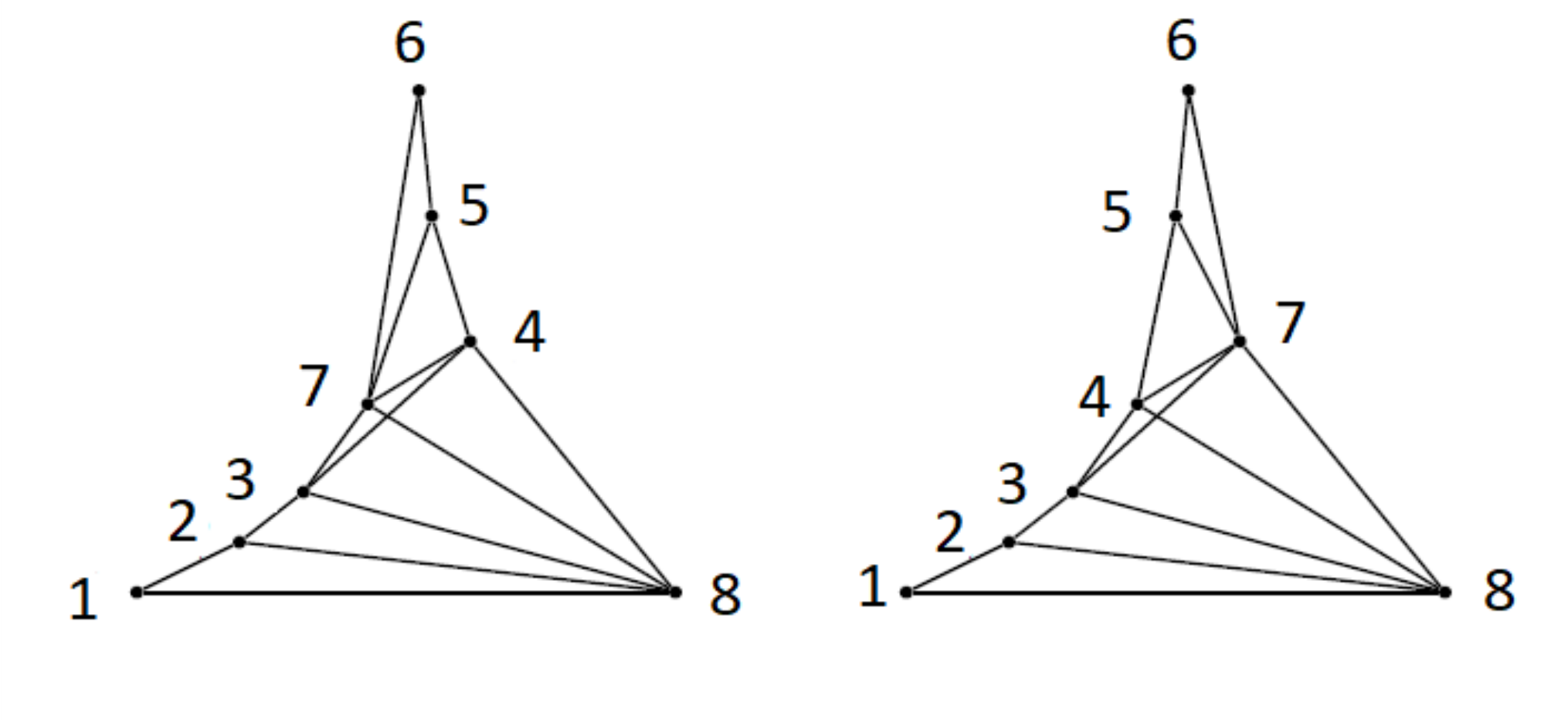}}
\vspace*{8pt}
\caption{Visibility graph of a tower polygon with two borderings.\label{fig:tower_border} }
\end{figure}

\begin{observation}
Vertices of a middle level are not on the same chain.
\end{observation}

\begin{observation}
\label{obs:bordering}
$G^\prime$ is a bipartite graph and for each edge $(p,q)\in E(G^\prime)$, $p$ and $q$ do not belong to the same chain.
\end{observation}

\begin{observation}
Any bordering that satisfies Observation~\ref{obs:bordering} has a realization as tower polygon. In this realization the order of vertices of each chain follows the order of the leveling numbers.
\end{observation}

Consequently, the Hamiltonian cycle of the realization of each bordering is unique and is computable in $O(|E|)$.

\begin{observation}
\label{obs:level_count}
A tower polygon with visibility graph $G(V,E)$ and leveling graph $G^\prime$ with $c$ connected components has at most 2 levelings and exactly $2^{c-1}$ borderings.
\end{observation}

\begin{figure}[ht]
\centerline{\includegraphics[scale=0.5]{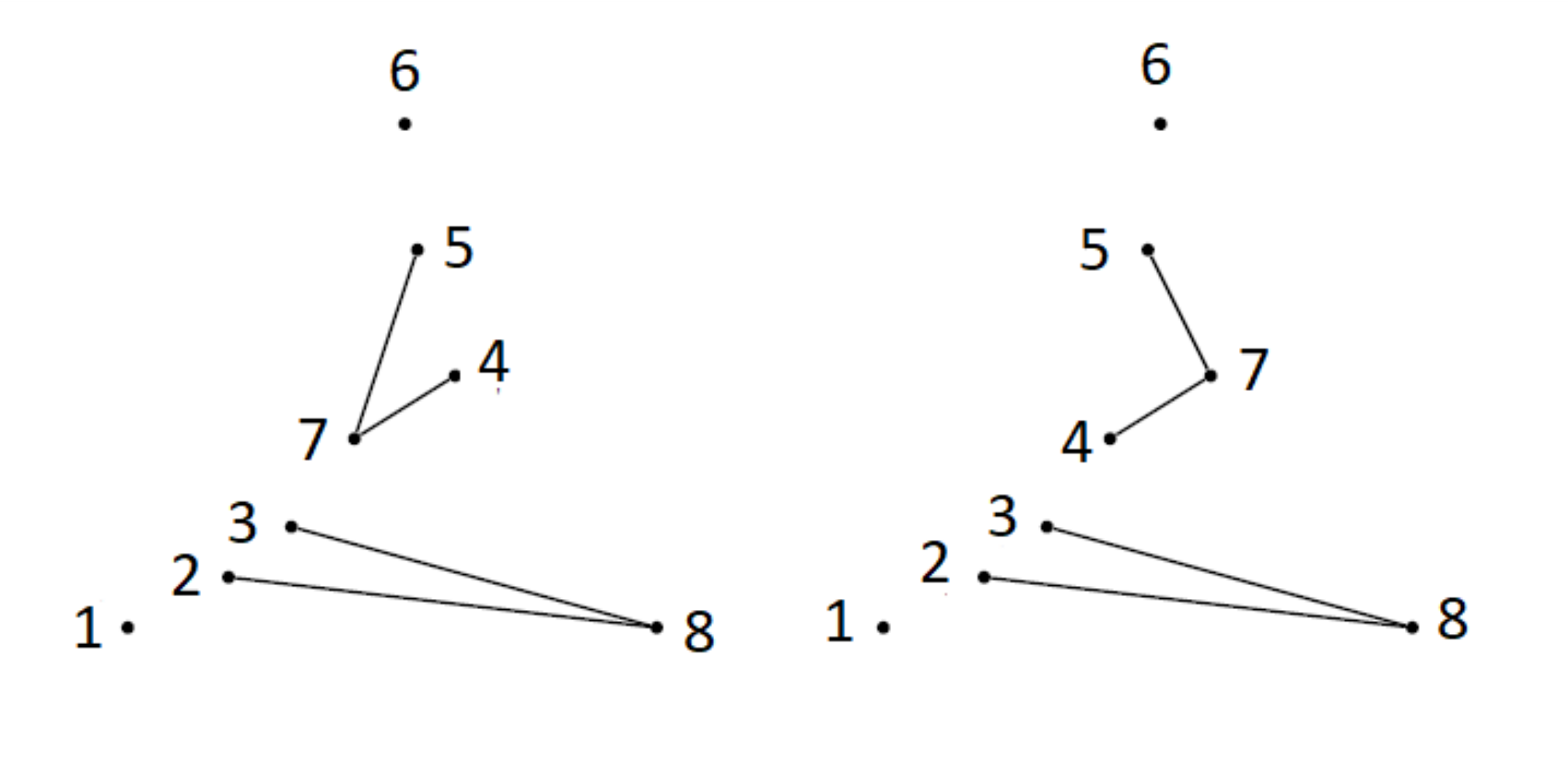}}
\vspace*{8pt}
\caption{Bordering graph of a tower polygon with two borderings.\label{fig:border_graph} }
\end{figure}

\subsection{Pseudo-tower polygons}
Consider a tower polygon in which some of the bottom vertices of one of its chains are removed in such a way that the last vertices of each chain are not visible to each other. Name this kind of polygons as pseudo-tower. According to this definition, a pseudo-tower polygon has no Hamiltonian cycle. This polygon is composed of a tower polygon and an induced path at the end of a chain which can not see any vertex from the other chain (see Figure~\ref{fig:inst1}). The induced path of a pseudo-tower polygons is called its \textit{tail}.

\begin{lemma}
\label{lem:pseudo-tower}
Leveling, bordering and boundary of a visibility graph $G(V,E)$ for a pseudo-tower polygon can be computed in $O(|E(G)|)$.
\end{lemma}
\begin{proof}
A visibility graph of a pseudo-tower polygon may have more than two vertices with two neighbours. But there is only one vertex of degree two such that its neighbours see each other. This vertex is the only candidate of the top vertex of pseudo-tower polygon. In addition, visibility graph of a pseudo-tower polygon has a single vertex $p$ with one neighbour which is the last vertex of the tail of the polygon. Therefore, we can start at this point and find the induced path in linear time. Removing the vertices of the induced path, leaves a visibility graph which corresponds to a tower polygon which can be reconstructed in linear time.
\end{proof}

\begin{figure}[ht]
\centerline{\includegraphics[scale=0.5]{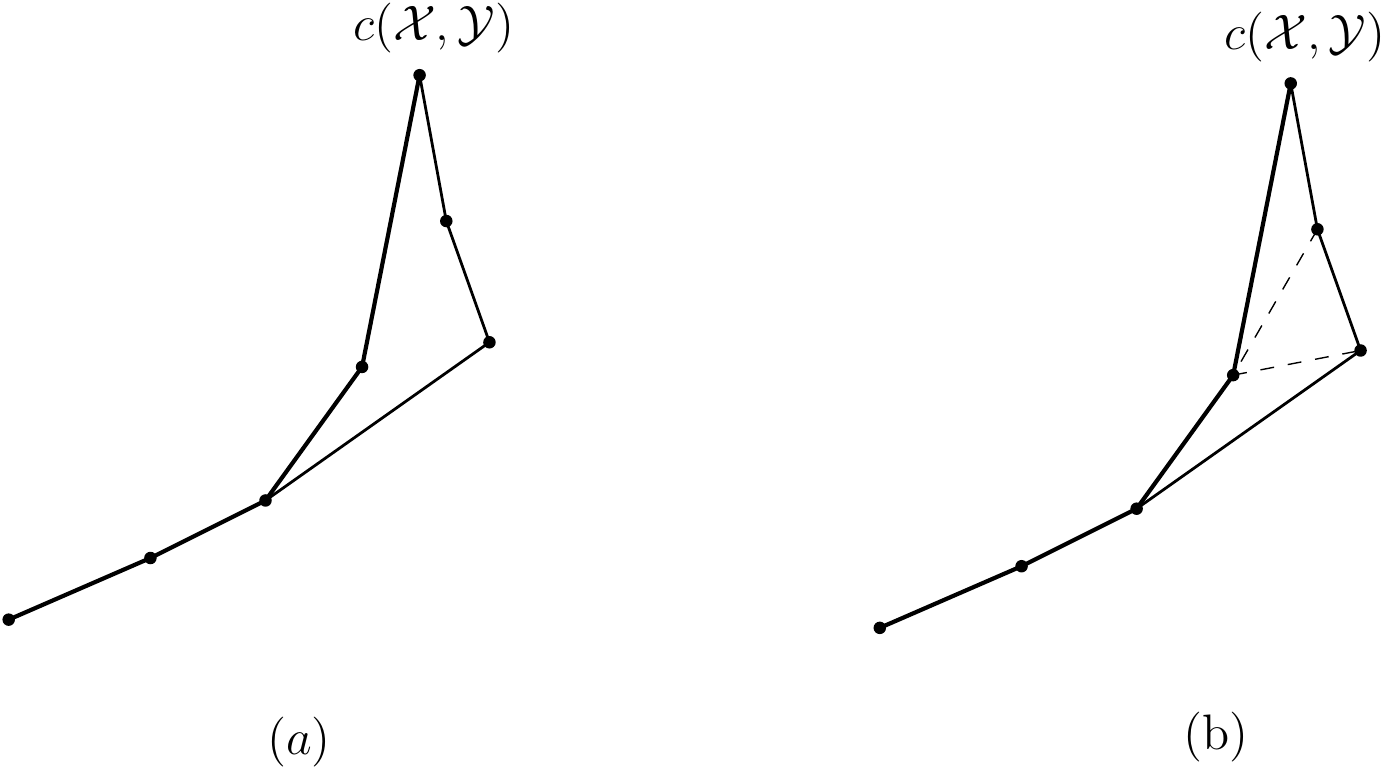}}
\vspace*{8pt}
\caption{a) Pseudo-tower polygon, b) The visibility graph of the polygon.\label{fig:inst1} }
\end{figure}

\subsection{Pseudo-triangle Polygons}
We assume that pseudo-triangles has the layout presented in Figure~\ref{fig:pseudo-triangle}-a with $\mathcal{U}$, $\mathcal{V}$ and $\mathcal{W}$ as respectively its left, right and bottom concave chains, and the common vertex between two concave chains like $\mathcal{V}$ and $\mathcal{W}$ is denoted by $c(\mathcal{V},\mathcal{W})$. The order and the name of vertices for the chains $\mathcal{U}$, $\mathcal{V}$, $\mathcal{W}$ are denoted by $<c(\mathcal{U},\mathcal{V}),u_1,u_2,...,c(\mathcal{U},\mathcal{W})>$, $<c(\mathcal{U},\mathcal{V}),v_1,v_2,...,c(\mathcal{V},\mathcal{W})>$ and $<c(\mathcal{U},\mathcal{W}),w_1,w_2,...,c(\mathcal{V},\mathcal{W})>$,  respectively. For a chain like $\mathcal{W}$ and a vertex $p$ on this chain, the $i_th$ vertex in the walk from $p$ toward $c(\mathcal{V},\mathcal{W})$ is denoted by $p^{i}_{c(\mathcal{U},\mathcal{W})}$. The set of all vertices on chain $\mathcal{U}$ which are visible from a vertex $p$ is denoted by $N_{\mathcal{U}}(p)$. The first and the last vertices on $\mathcal{U}$ in the walk from $c(\mathcal{U},\mathcal{V})$ toward $c(\mathcal{U},\mathcal{W})$, which is visible to all vertices of the set $S={p_1,...,p_i}$, are denoted by $u_{c(\mathcal{U},\mathcal{V})}(S)$ and $U_{c(\mathcal{U},\mathcal{V})}(S)$, respectively. In addition, we assume that that the top joint vertex $c(\mathcal{U},\mathcal{V})$ has the least degree between the joint vertices.

\begin{figure}[ht]
\centerline{\includegraphics[scale=0.7]{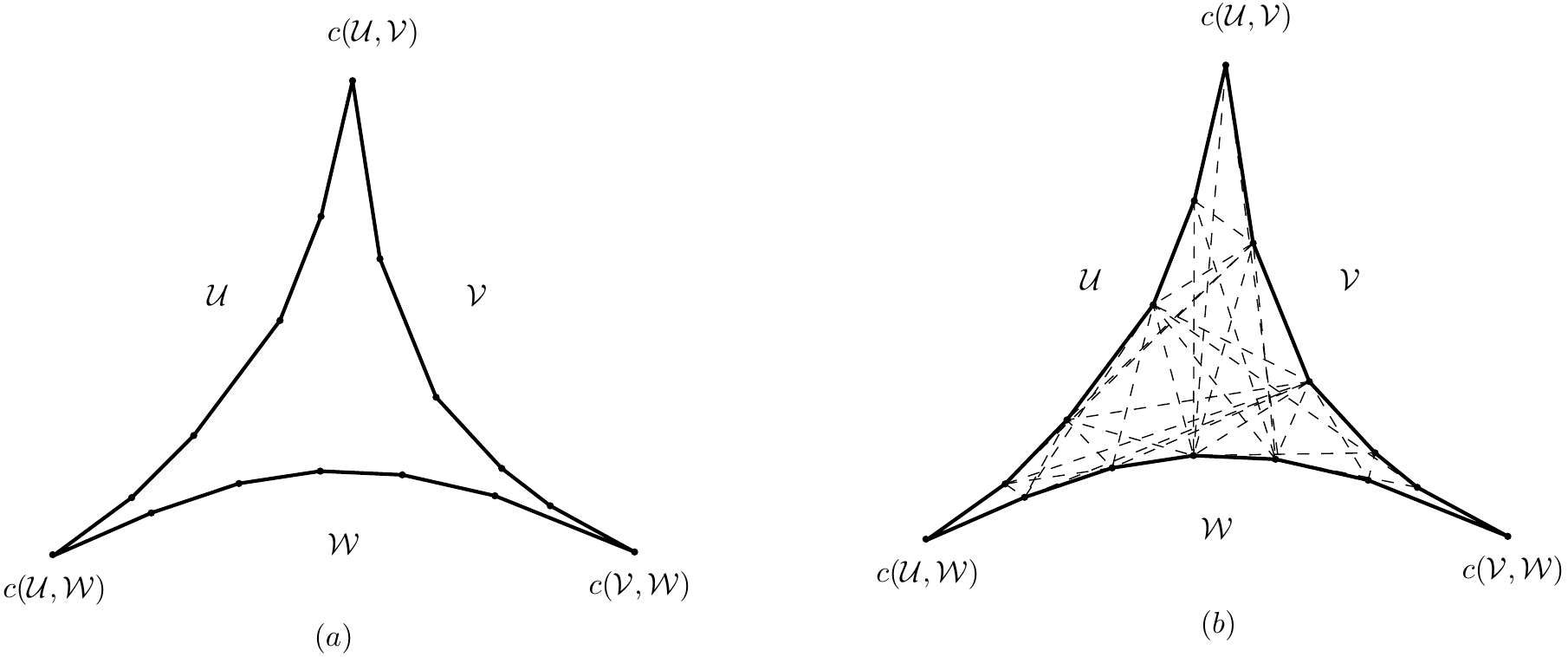}}
\vspace*{8pt}
\caption{a) A pseudo-triangle, b) The visibility graph of the polygon.\label{fig:pseudo-triangle} }
\end{figure}

\section{Computing Hamiltonian cycle}
\label{sec:C}
For a given graph $G$, we present our method to find the Hamiltonian cycle $H(P)$ corresponding to the boundary cycle of some pseudo-triangle $P$ whose visibility graph is $G$. We first find the vertices which can be candidates for the top vertex $c(\mathcal{U},\mathcal{V})$. Then, we use this vertex to split the visibility graph into some regions. Finally, we introduce some necessary constraints in the visibility graph of pseudo-triangles and the detected regions, to be used to extract the Hamiltonian cycle from the visibility graph.

\subsection{Find the top join vertex $c(\mathcal{U},\mathcal{V})$}
As assumed before, $c(\mathcal{U},\mathcal{V})$ has the least degree between the joint vertices.

\begin{lemma}
\label{lem:top-degree}
Assume that $p$ and $q$ are the extreme(joint vertices) of a chain of a pseudo-triangle $P$ and $v$ is a non-joint vertex on this chain. The degree of $v$ in the visibility graph of $P$ is strictly less than the degree of exactly one of the vertices $p$ and $q$.
\end{lemma}
\begin{proof}
Without lose of generality, assume that $p$ is a vertex on chain $\mathcal{U}$. Either all of the vertices which are visible from $c(\mathcal{U},\mathcal{V})$ are visible from $p$ or some vertices on chain $\mathcal{W}$ blocked their visiblity. In the latter case, no vertex can block the visibility of $p$ and the vertices which are visible from $c(\mathcal{U},\mathcal{W})$. In addition each of the vertices which is adjacent to $p$ on the Hamiltonian cycle is either invisible to $c(\mathcal{U},\mathcal{V})$ or $c(\mathcal{U},\mathcal{W})$. Therefore, the degree of $p$ is strictly higher than one of the endpoints of its chain.
\end{proof}

The direct result of Lemma~\ref{lem:top-degree} is that the degree of each non-joint vertex of a pseudo-triangle is strictly higher than one of the joint vertices of its chain. Therefore, the minimum degree of the vertices of the graph belongs to one of the joint vertices. Remind that we assume that the joint vertex with the least minimum degree in $G(P)$ is the top joint vertex of $P$. The degree of this vertex is strictly less than all vertices of $P$ in $G(P)$. Therefore, to find the top vertex of $P$ we need to find the vertex with minimum degree in $G(P)$.

\begin{theorem}
\label{thm:top}
In a visibility graph $G(P)$ of a pseudo-triangle $P$, there are at most three candidates for $c(\mathcal{U},\mathcal{V})$ which are all joint vertices.
\end{theorem}
\begin{proof}
As the degree of the joint vertex ($\delta(G)$) in $P$ is strictly less than all other non-joint vertices, $G(P)$ have at most 3 vertex with degree equal to $\delta(G)$ which are the joint vertices of $P$.
\end{proof}

\subsection{Split the Polygon}
\label{sec:split}
It is simple to see that there is always a vertex $w^\prime_0$ on $\mathcal{W}$ which is visible to vertices from both chains $\mathcal{U}$ and $\mathcal{V}$. We assume that there is another vertex $w^\prime_1$ on $\mathcal{W}$ and adjacent to $w^\prime_0$ with this property. The degenerate cases where there is only a single vertex which is visible from both chains $\mathcal{U}$ and $\mathcal{V}$ will be handled separately in Section~\ref{sec:degen}. Then, the edge $e=(w^\prime_0, w^\prime_1)$ is called a split-edge. Assume that this special edge is known. The vertices of $P$ which are above this edge correspond to a tower polygon and we denote this polygon as $C_e(G)$. By removing the vertices of $C_e(G)$ from $G$, the rest of the graph will be splited into two connected components. Both of these connected components are also pseudo-tower polygons. Lets name the one which contains $c(\mathcal{U},\mathcal{W})$ as $A_e(G)$ and the other one, which contains $c(\mathcal{V},\mathcal{W})$, as $B_e(G)$ (see Figure~\ref{fig:split}). Note that although we have defined these parts based on the realization of the pseudo-triangle, but, their combinatorial structures only depend on the visibility graph and the edge $e=(w^\prime_0, w^\prime_1)$.

\subsection{Finding the vertices of $C_e(G)$}
\label{sec:topC}

There are two types of vertices in $C_e(G)$:
\begin{enumerate}
\item Vertices which are visible to both $w^\prime_0$ and $w^\prime_1$. These vertices are the only vertices in $G(P)$ which have this property and we can find them in linear time. If $c(\mathcal{U},\mathcal{V})$ belongs to this group, all vertices of $C_e(G)$ belongs to this group and the there will be no vertex in the next group.
\item Vertices which are not visible to either $w^\prime_0$ or $w^\prime_1$. These vertices are placed above the vertices of the previous group. These vertices can be determined in linear time by using the leveling algorithm for tower polygons. 
\end{enumerate}

Name the set of the vertices which are visible to both vertices of $e$ as $X$, and the set of vertices which are invisible to one of the endpoints of $e$, $Y$. A vertex in pseudo-triangle polygon can see one connected part of each concave chain. Beside that, the vertices of the last level of $C_e(G)$ are visible to both vertices of $e$. Therefore, the leveling algorithm traverse and detect all vertices of $Y$ before the vertices of $X$ and we can start leveling from the top joint vertex and continue the leveling process until the leveling reach to a vertex which is visible to both vertices of $e$. Name this set of vertices as $C_1$. Adding all vertices which are visible to both end points of $e$, to the set $C_1$ makes the set of vertices of $C_e(G)$.

\begin{figure}[ht]
\centerline{\includegraphics[scale=0.75]{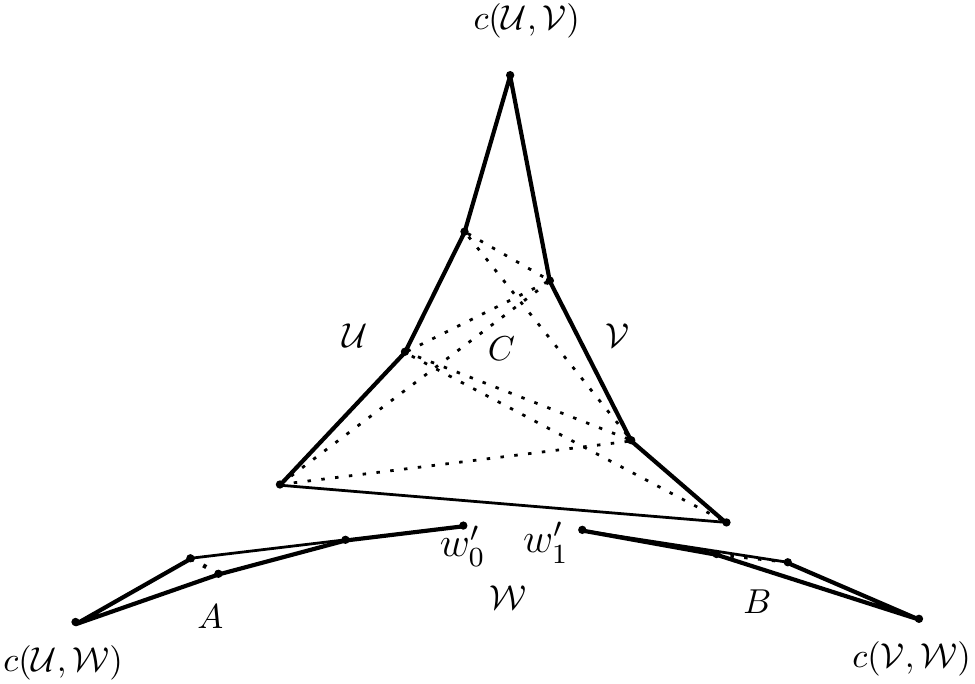}}
\vspace*{8pt}
\caption{Splited pseudo-triangle with respect to $(w^\prime_0, w^\prime_1)$.\label{fig:split} }
\end{figure}

After determining the vertices of $C_e(G)$ and removing them from $G(P)$, the vertices of $A_e(G)$ and $B_e(G)$ are the connected components of the remained part of $G{P}$.

\subsection{Degenerate Cases}
\label{sec:degen}

The degenerate case of pseudo-triangle polygons is the case in which there is no split-edge (there is only one vertex $w^\prime_0$) on the chain $\mathcal{W}$ which is visible from some vertices of both of the other chains. In this case we need another method to split the polygon. Similar to the method discussed in Section~\ref{sec:topC}, we choose a split edge $e=w^\prime_0w^\prime_1$ in $G(P)$ and split the pseudo-triangle polygon into three parts $A_e(G)$, $B_e(G)$ and $C_e(G)$.  We assume that $w^\prime_1$ is $w^{\prime 1}_{0 c(\mathcal{V},\mathcal{W})}$. Therefore $C_e(G)$ which contains the vertices above this split-edge contains no vertex in chain $\mathcal{V}$. The rest of process to find vertices of $C_e(G)$ and find $A_e(G)$ and $B_e(G)$ remains unchanged.

\subsection{Hamiltonian cycle of the splited parts}
As stated before, all three parts, $A_e(G)$, $B_e(G)$, and $C_e(G)$, are tower or pseudo-tower polygons. Therefore, their recognition, reconstruction and finding their Hamiltonian cycles can be solving independently. However, Observation~\ref{obs:level_count} states that there is no unique solution for the Hamiltonian cycle of a tower polygon. In addition, there are also some visibility constraints in $G(P)$ between the vertices of $A_e(G)$, $B_e(G)$, and $C_e(G)$ which have not been taken into consideration, yet. These constrains (bordering constraints) are discussed with more details as follows.

\subsection{Split constraints}
The bordering graph $C_e^\prime(G)$, which corresponds to the sub-polygon $C_e(G)$, can be computed by applying the leveling method on this visibility graph. As stated in Lemma~\ref{obs:bordering} and Lemma~\ref{obs:level_count}, $C_e^\prime(G)$ is a bipartite graph with some connected components. Without considering the visibility relations between vertices of $C_e(G)$ and the chain $\mathcal{W}$, each partition of these connected components can be arbitrarily assigned to each of the chains $\mathcal{U}$ and $\mathcal{V}$. Name the last vertex of $\mathcal{U}$ as $p$ and the last vertex of $\mathcal{V}$ as $q$. Denote the chain of vertices of $\mathcal{W}$ which are visible from both of $p$ and $q$ as \textit{common-chain}. There are two candidates for $p^{1}_{c(\mathcal{U},\mathcal{W})}$ in $A_e(G)$ (similarly for $q^{1}_{c(\mathcal{V},\mathcal{W})}$ in $B_e(G)$). Denote the chain of vertices of $\mathcal{W}$ which are visible from $p^{1}_{c(\mathcal{V},\mathcal{W})}$ (resp. $p^{1}_{c(\mathcal{V},\mathcal{W})}$) and are not on the common-chain as $W^\prime$ (resp. $W^{''}$). In addition, denote the chain which is composed of vertices of $W^\prime$, $W^{''}$ as \textit{split-chain}. The following lemmas introduce two visibility constraints on bordering the vertices of $C_e(G)$ (bordering constraints):

\begin{lemma}
\label{lem:sidevis}
Vertices of $\mathcal{V}$ are invisible to vertices of $W^\prime$. Similarly, vertices of $\mathcal{U}$ are invisible to the vertices of $W^{''}$.
\end{lemma}
\begin{proof}
If a vertex of $W^{\prime}$ or $W^{''}$ violates this condition, by the definition of the common-chain it belongs to common-chain.
\end{proof}

\begin{lemma}
\label{lem:downvis}
For two vertices $a$ and $b$ on the chain $\mathcal{V}$ (resp. $\mathcal{U}$) of $C_e(G)$ such that the level of $b$ is bigger than $a$, the set of vertices of $B_e(G)$ (resp. $A_e(G)$) which are visible from $b$ are a super set of the set of vertices of $B_e(G)$ (resp. $A_e(G)$) which are visible from $a$.
\end{lemma}
\begin{proof}
The vertices of the chain $\mathcal{V}$ and the set of vertices of $\mathcal{W}$ which are visible from some vertex in $\mathcal{V}$ make a pseudo-tower polygon. The statement of the lemma holds in every pseudo-tower polygon.
\end{proof}

For a vertex $a$ on chain $\mathcal{U}$ in $C_e(G)$ the vertex with the highest level in $\mathcal{V}$ is the vertex which blocked the sight of $a$ and the invisible vertices (blocking vertex of $a$) of $W^{''}$. Therefore, we have:

\begin{lemma}
\label{lem:blockvis}
For every vertex $a$ on chain $\mathcal{U}$ (resp. $\mathcal{U}$) in $C_e(G)$, the set of vertices of $B_e(G)$ (resp. $A_e(G)$) which are visible to the blocking vertex of $a$ is a super set of the set of vertices of $B_e(G)$ (resp. $A_e(G)$) which are visible to $a$.
\end{lemma}

Verifying each of these constraints for a boardering needs $O(E)$ time. In addition these constraints force that two different borderings which satisfy these constraints to be isomorph to each other, because if a pair of vertices are arbitrarily placed on $\mathcal{V}$ or $\mathcal{W}$ then they must have the same neighbours in $A_e(G)$ and $B_e(G)$ and also the same neighbours in split-chain (otherwise, swapping the chain of the vertices violates one of the increasing pattern in either Lemma~\ref{lem:downvis} or Lemma~\ref{lem:blockvis}). Therefore, there is only constant number of non-isomorph borderings which satisfy these three linear time verifiable constraints. To find this bordering we can begin with an arbitrary bordering of $C_e(G)$ and sweep the vertices of one of its chains from top to buttom. If a constraint is violated on a vertex like $a$, the chain of vertices of the bordering graph, $G^\prime$, which are in the same component with $a$ must be swapped so that the bordering graph remains bipartite and have a chance to satisfy the bordering constraints. By continueing these constraints, it will either iterate over all vertices of $C_e(G)$ or it violates one of the constraints. Therefore we have the following theorem:

\begin{theorem}
\label{thm:split_consts}
There are linear time verifiable constraints (bordering constraints) such that $O(1)$ number of non-isomorph pairs $(G,H)$ for a pseudo-triangle with visibility graph $G$, Hamiltonian cycle $H$, the top vertex $c(\mathcal{U},\mathcal{V})$ and the split-edge $e$, satisfies them.
\end{theorem}

\subsection{Recognizing the visibility graph}
Theorem~\ref{thm:split_consts} leads to the final step of our method. We have a visibility graph $G$ with a Hamiltonian cycle $H$ and we need to verify that whether there is a pseudo-triangle with this pair as its visibility graph and Hamiltonian cycle. As stated in Section~\ref{sec:A}, this problem can be solved in $O(E)$. In the next section we review the outline and computational complexity of our method and prove that all the steps are polynomial time computable. 

\section{Computational complexity}
Our method to find Hamiltonian cycle of a visibility graph $G$ of a pseudo-triangle has these steps:
\begin{enumerate}
\item Determine the candidates of the top joint vertex of the polygon. Lemma~\ref{lem:top-degree} states that this step takes $O(E)$ time and obtains at most three candidates for the top joint vertex. 
\item Determine the split-edge of the polygon. This step chooses at most $O(E)$ candidates for split-edge.
\item Split the polygon using the determined split-edge. In this step, we find the vertices of the sub-polygon $C$ and split the rest of the graph into its two connected components. The induced graph on the vertices of each connected component produces the sub-polygons $B$ and $C$. This can be computed in $O(E)$ time.
\item Compute the leveling and $G^\prime$ of each sub-polygon and determine the Hamiltonian cycle of each sub-polygon. The leveling of each sub-polygon is computable in $O(E)$ time using the leveling algorithm for tower polygons.
\item Determine the split-chain. There are $O(1)$ choices for the last vertex of each of the chains of each sub-polygon. This takes $O(E)$ time to determine these candidates from the leveling graph. After determining these vertices, it takes $O(E)$ time to find the complete split-chain.
\item Compute the Hamiltonian cycle of the graph. It takes $O(E)$ time to apply the constraints on the leveling graph and find the Hamiltonian cycle (bordering) of the pseudo-triangle.
\item Verify the pair of visibility graph and Hamiltonian cycle. It takes $O(E)$ time to solve the recognition problem of pseudo-triangle for a pair of visibility graph and Hamiltonian cycle.
\end{enumerate}

In summary, this method takes a graph $G$ and spends $O(E^2)$ time to determine $O(E)$ pairs of $(G,H)$ as the total possible candidates for the visibility graph and Hamiltonian cycle of a pseudo-triangle whose visibility graph and $G$ are isomorph. It takes $O(E)$ time to solve the recognition problem for each of these pairs. Then, we have the final theorem:

\begin{theorem}
Having only the visibility graph $G(V,E)$, recognizing and reconstruction problems for pseudo-triangles can be solved in $O(E^2)$.
\end{theorem}


\small
\bibliographystyle{abbrv}

\bibliography{main}



\end{document}